\documentclass[sigconf]{acmart}

\usepackage{booktabs} 

\setcopyright{rightsretained}

\acmConference[KDD 2017]{ACM SIGKDD Conference on Knowledge Discovery and Data Mining}{August 2017}{Halifax, Nova Scotia, Canada} 
\acmYear{2017}
\copyrightyear{2017}

\acmPrice{15.00}

\usepackage[ruled, vlined, nofillcomment, linesnumbered]{algorithm2e}
\usepackage{graphicx}
\usepackage{subfigure}
\usepackage{url}
\usepackage{booktabs}
\usepackage{array}
\usepackage{epsfig}
\usepackage{color}
\usepackage{comment}
\usepackage{mathrsfs}
\usepackage{amsmath,amssymb,amsthm} 
\usepackage{psfrag}
\usepackage{fullpage}
\usepackage{balance}  

\usepackage{tikz}
\usetikzlibrary{fit,external,positioning}
\usepackage[latin1]{inputenc}
\usepackage[T1]{fontenc}

\newcommand{\set}[1]{\left\{#1\right\}}
\newcommand{\pr}[1]{\left(#1\right)}

\newcommand{\abs}[1]{{\left|#1\right|}}

\newcommand{\enset}[2]{\left\{#1 ,\ldots , #2\right\}}

\newcommand{\np}{\textbf{NP}}

\newcommand{\define}{\leftarrow}

\newcommand{\prbminstc}{\textsc{MinViol}\xspace}
\newcommand{\prbmaxtri}{\textsc{MaxTri}\xspace}

\newtheorem{problem}{Problem}

\newcommand{\argmax}{\operatornamewithlimits{arg\,max}}

\newcommand{\NP}{\ensuremath{\mathbf{NP}}}
\newcommand{\Poly}{\ensuremath{\mathbf{P}}}

\newcommand{\violations}{\ensuremath{\mathit{viol}}}
\newcommand{\covered}{\ensuremath{\mathit{tri}}}

\newcommand{\bigO}{\ensuremath{O}}

\newcommand{\stc}{{\sc\Large stc}\xspace}

\newcommand{\spara}[1]{\smallskip\noindent{\bf{#1}}}

\newcommand{\squishlist}{\begin{list}{$\bullet$}
  { \setlength{\itemsep}{0pt}
     \setlength{\parsep}{3pt}
     \setlength{\topsep}{3pt}
     \setlength{\partopsep}{0pt}
     \setlength{\leftmargin}{1.5em}
     \setlength{\labelwidth}{1em}
     \setlength{\labelsep}{0.5em} } }
\newcommand{\squishend}{
  \end{list}  }

\newcommand{\avg}{\ensuremath{\mathrm{avg}}}
\newcommand{\edges}{\ensuremath{E}}
\newcommand{\indedges}{\ensuremath{E_0}}

\newcommand{\DBLP}{{\textsl{DBLP}}\xspace}
\newcommand{\youtube}{{\textsl{Youtube}}\xspace}

\newcommand{\KDD}{{\textsl{KDD}}\xspace}
\newcommand{\ICDM}{{\textsl{ICDM}}\xspace}
\newcommand{\FBcirc}{{\textsl{FB-circles}}\xspace}
\newcommand{\FBfeat}{{\textsl{FB-features}}\xspace}
\newcommand{\lastFMart}{{\textsl{lastFM-artists}}\xspace}
\newcommand{\lastFMtag}{{\textsl{lastFM-tags}}\xspace}
\newcommand{\DBbook}{{\textsl{DB-bookmarks}}\xspace}
\newcommand{\DBtag}{{\textsl{DB-tags}}\xspace}

\newcommand{\BLS}{{\textsl{Sintos}}\xspace}
\newcommand{\BLA}{{\textsl{Angluin}}\xspace}
\newcommand{\our}{{\textsl{Greedy}}\xspace}

\newcommand{\BLSalt}{\mathit{S}\xspace}
\newcommand{\BLAalt}{\mathit{A}\xspace}

\pgfdeclarelayer{background}
\pgfdeclarelayer{foreground}
\pgfsetlayers{background,main,foreground}

\definecolor{yafaxiscolor}{rgb}{0.3, 0.3, 0.3}
\definecolor{yafcolor1}{rgb}{0.4, 0.165, 0.553}
\definecolor{yafcolor2}{rgb}{0.949, 0.482, 0.216}
\definecolor{yafcolor3}{rgb}{0.47, 0.549, 0.306}
\definecolor{yafcolor4}{rgb}{0.925, 0.165, 0.224}
\definecolor{yafcolor5}{rgb}{0.141, 0.345, 0.643}
\definecolor{yafcolor6}{rgb}{0.965, 0.933, 0.267}
\definecolor{yafcolor7}{rgb}{0.627, 0.118, 0.165}
\definecolor{yafcolor8}{rgb}{0.878, 0.475, 0.686}

\begin{document}
\title{Inferring the strength of social ties:\\a community-driven approach}

\author{Polina Rozenshtein}
\affiliation{%
  \institution{HIIT, Aalto University}
  \city{Espoo} 
  \country{Finland} 
}
\email{polina.rozenshtein@aalto.fi}

\author{Nikolaj Tatti}
\affiliation{%
  \institution{HIIT, Aalto University}
  \city{Espoo} 
  \country{Finland} 
}
\email{nikolaj.tatti@aalto.fi}

\author{Aristides Gionis}
\affiliation{%
  \institution{HIIT, Aalto University}
  \city{Espoo} 
  \country{Finland} 
}
\email{aristides.gionis@aalto.fi}

\begin{abstract}
Online social networks are growing and becoming denser. 
The social connections of a given person may have very high variability: 
from close friends and relatives to acquaintances to people who hardly know.
Inferring the strength of social ties is an important ingredient 
for modeling the interaction of users in a network and understanding their behavior. 
Furthermore, the problem has applications in computational social science, 
viral marketing, and people recommendation.

In this paper we study the problem of inferring the strength of social ties
in a given network.
Our work is motivated by a recent approach~\cite{sintos2014using}, 
which leverages the {\em strong triadic closure} (\stc) principle, 
a hypothesis rooted in social psychology~\cite{granovetter1973strength}.
To guide our inference process, 
in addition to the network structure, 
we also consider as input a collection of {\em tight} communities.
Those are sets of vertices that we expect to be connected via strong ties. 
Such communities appear in different situations, e.g.,
when being part of a community implies a strong connection to one of the existing members.

We consider two related problem formalizations
that reflect the assumptions of our setting: 
small number of \stc violations and strong-tie connectivity in the input communities. 
We show that both problem formulations are \NP-hard. 
We also show that one problem formulation is hard to approximate, 
while for the second we develop an algorithm with approximation guarantee.
We validate the proposed method on real-world datasets
by comparing with baselines that optimize 
\stc violations and community connectivity separately.

\end{abstract}



\settopmatter{printacmref=false, printfolios=false}
\setcopyright{none}
\pagestyle{plain} 

\maketitle

\section{Introduction}
\label{sec:intro}

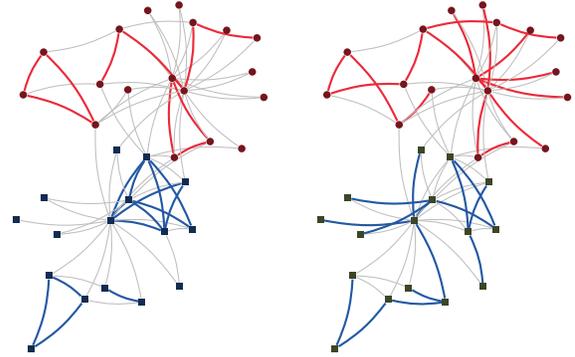
\begin{figure}
\begin{tikzpicture}[scale = 1]
\tikzstyle{lm1} = [fill = yafcolor5!50!black, inner sep = 1.2pt]
\tikzstyle{lm2} = [circle, fill = yafcolor4!50!black, inner sep = 1pt]
\tikzstyle{ll} = []
\tikzstyle{le1} = [yafcolor5, thick, bend left = 10]
\tikzstyle{le2} = [yafcolor4, thick, bend left = 10]
\tikzstyle{le3} = [gray!50, bend left = 10]
\node[lm1] (n0) at (-0.39521, -0.90369) {};
\node[lm1] (n1) at (-0.1494, -0.63131) {};
\node[lm1] (n2) at (0.089756, -0.058231) {};
\node[lm1] (n3) at (0.32331, -1.0532) {};
\node[lm1] (n4) at (0.023263, -1.9859) {};
\node[lm1] (n5) at (-1.2108, -1.6373) {};
\node[lm1] (n6) at (-0.73841, -1.9544) {};
\node[lm1] (n7) at (0.69671, -1.0172) {};
\node[lm2] (n8) at (0.45414, -0.064462) {};
\node[lm2] (n9) at (1.35, 0.053898) {};
\node[lm1] (n10) at (-0.47034, -1.7978) {};
\node[lm1] (n11) at (-1.6432, -0.89285) {};
\node[lm1] (n12) at (0.52145, -1.7831) {};
\node[lm1] (n13) at (0.59789, -0.39239) {};
\node[lm2] (n14) at (1.15, 1.6278) {};
\node[lm2] (n15) at (1.644, 0.73643) {};
\node[lm1] (n16) at (-1.4464, -2.6101) {};
\node[lm1] (n17) at (-1.278, -0.60231) {};
\node[lm2] (n18) at (0.51654, 1.9636) {};
\node[lm1] (n19) at (-0.30691, 0.035621) {};
\node[lm2] (n20) at (0.10109, 1.891) {};
\node[lm1] (n21) at (-1.1086, -1.0869) {};
\node[lm2] (n22) at (1.4907, 1.075) {};
\node[lm2] (n23) at (-0.2767, 1.6423) {};
\node[lm2] (n24) at (-1.5541, 0.77017) {};
\node[lm2] (n25) at (-1.2836, 1.3377) {};
\node[lm2] (n26) at (1.5531, 1.5267) {};
\node[lm2] (n27) at (-0.5343, 0.9102) {};
\node[lm2] (n28) at (-0.16321, 0.83704) {};
\node[lm2] (n29) at (0.70218, 1.7305) {};
\node[lm2] (n30) at (0.93024, 0.14795) {};
\node[lm2] (n31) at (-0.5944, 0.37166) {};
\node[lm2] (n32) at (0.42538, 0.99008) {};
\node[lm2] (n33) at (0.58393, 0.82341) {};
\begin{pgfonlayer}{background}
\draw (n0) edge[le1] (n1);
\draw (n0) edge[le1] (n2);
\draw (n0) edge[le1] (n3);
\draw (n0) edge[le3] (n4);
\draw (n0) edge[le3] (n5);
\draw (n0) edge[le3] (n6);
\draw (n0) edge[le3] (n7);
\draw (n0) edge[le3] (n8);
\draw (n0) edge[le3] (n10);
\draw (n0) edge[le3] (n11);
\draw (n0) edge[le3] (n12);
\draw (n0) edge[le1] (n13);
\draw (n0) edge[le3] (n17);
\draw (n0) edge[le3] (n19);
\draw (n0) edge[le3] (n21);
\draw (n0) edge[le3] (n31);
\draw (n1) edge[le1] (n2);
\draw (n1) edge[le1] (n3);
\draw (n1) edge[le1] (n7);
\draw (n1) edge[le3] (n13);
\draw (n1) edge[le3] (n17);
\draw (n1) edge[le3] (n19);
\draw (n1) edge[le3] (n21);
\draw (n1) edge[le3] (n30);
\draw (n2) edge[le1] (n3);
\draw (n2) edge[le1] (n7);
\draw (n2) edge[le3] (n8);
\draw (n2) edge[le3] (n9);
\draw (n2) edge[le3] (n13);
\draw (n2) edge[le3] (n27);
\draw (n2) edge[le3] (n28);
\draw (n2) edge[le3] (n32);
\draw (n3) edge[le3] (n7);
\draw (n3) edge[le3] (n12);
\draw (n3) edge[le1] (n13);
\draw (n4) edge[le3] (n6);
\draw (n4) edge[le1] (n10);
\draw (n5) edge[le1] (n6);
\draw (n5) edge[le3] (n10);
\draw (n5) edge[le1] (n16);
\draw (n6) edge[le1] (n16);
\draw (n8) edge[le2] (n30);
\draw (n8) edge[le2] (n32);
\draw (n8) edge[le3] (n33);
\draw (n9) edge[le3] (n33);
\draw (n13) edge[le3] (n33);
\draw (n14) edge[le3] (n32);
\draw (n14) edge[le3] (n33);
\draw (n15) edge[le3] (n32);
\draw (n15) edge[le3] (n33);
\draw (n18) edge[le3] (n32);
\draw (n18) edge[le3] (n33);
\draw (n19) edge[le3] (n33);
\draw (n20) edge[le3] (n32);
\draw (n20) edge[le3] (n33);
\draw (n22) edge[le3] (n32);
\draw (n22) edge[le3] (n33);
\draw (n23) edge[le3] (n25);
\draw (n23) edge[le2] (n27);
\draw (n23) edge[le3] (n29);
\draw (n23) edge[le3] (n32);
\draw (n23) edge[le2] (n33);
\draw (n24) edge[le2] (n25);
\draw (n24) edge[le3] (n27);
\draw (n24) edge[le2] (n31);
\draw (n25) edge[le2] (n31);
\draw (n26) edge[le2] (n29);
\draw (n26) edge[le3] (n33);
\draw (n27) edge[le3] (n33);
\draw (n28) edge[le3] (n31);
\draw (n28) edge[le3] (n33);
\draw (n29) edge[le3] (n32);
\draw (n29) edge[le2] (n33);
\draw (n30) edge[le2] (n32);
\draw (n30) edge[le3] (n33);
\draw (n31) edge[le3] (n32);
\draw (n31) edge[le3] (n33);
\draw (n32) edge[le2] (n33);
\end{pgfonlayer}
\end{tikzpicture}
\hspace{0.5cm}
\begin{tikzpicture}[scale = 1]
\tikzstyle{lm1} = [fill = yafcolor3!50!black, inner sep = 1.2pt]
\tikzstyle{lm2} = [circle, fill = yafcolor4!50!black, inner sep = 1pt]
\tikzstyle{ll} = []
\tikzstyle{le1} = [yafcolor5, thick, bend left = 10]
\tikzstyle{le2} = [yafcolor4, thick, bend left = 10]
\tikzstyle{le3} = [gray!50, bend left = 10]
\node[lm1] (n0) at (-0.39521, -0.90369) {};
\node[lm1] (n1) at (-0.1494, -0.63131) {};
\node[lm1] (n2) at (0.089756, -0.058231) {};
\node[lm1] (n3) at (0.32331, -1.0532) {};
\node[lm1] (n4) at (0.023263, -1.9859) {};
\node[lm1] (n5) at (-1.2108, -1.6373) {};
\node[lm1] (n6) at (-0.73841, -1.9544) {};
\node[lm1] (n7) at (0.69671, -1.0172) {};
\node[lm2] (n8) at (0.45414, -0.064462) {};
\node[lm2] (n9) at (1.35, 0.053898) {};
\node[lm1] (n10) at (-0.47034, -1.7978) {};
\node[lm1] (n11) at (-1.6432, -0.89285) {};
\node[lm1] (n12) at (0.52145, -1.7831) {};
\node[lm1] (n13) at (0.59789, -0.39239) {};
\node[lm2] (n14) at (1.15, 1.6278) {};
\node[lm2] (n15) at (1.644, 0.73643) {};
\node[lm1] (n16) at (-1.4464, -2.6101) {};
\node[lm1] (n17) at (-1.278, -0.60231) {};
\node[lm2] (n18) at (0.51654, 1.9636) {};
\node[lm1] (n19) at (-0.30691, 0.035621) {};
\node[lm2] (n20) at (0.10109, 1.891) {};
\node[lm1] (n21) at (-1.1086, -1.0869) {};
\node[lm2] (n22) at (1.4907, 1.075) {};
\node[lm2] (n23) at (-0.2767, 1.6423) {};
\node[lm2] (n24) at (-1.5541, 0.77017) {};
\node[lm2] (n25) at (-1.2836, 1.3377) {};
\node[lm2] (n26) at (1.5531, 1.5267) {};
\node[lm2] (n27) at (-0.5343, 0.9102) {};
\node[lm2] (n28) at (-0.16321, 0.83704) {};
\node[lm2] (n29) at (0.70218, 1.7305) {};
\node[lm2] (n30) at (0.93024, 0.14795) {};
\node[lm2] (n31) at (-0.5944, 0.37166) {};
\node[lm2] (n32) at (0.42538, 0.99008) {};
\node[lm2] (n33) at (0.58393, 0.82341) {};
\begin{pgfonlayer}{background}
\draw (n0) edge[le1] (n1);
\draw (n0) edge[le3] (n2);
\draw (n0) edge[le3] (n3);
\draw (n0) edge[le1] (n4);
\draw (n0) edge[le3] (n5);
\draw (n0) edge[le3] (n6);
\draw (n0) edge[le3] (n7);
\draw (n0) edge[le3] (n8);
\draw (n0) edge[le3] (n10);
\draw (n0) edge[le1] (n11);
\draw (n0) edge[le3] (n12);
\draw (n0) edge[le3] (n13);
\draw (n0) edge[le3] (n17);
\draw (n0) edge[le1] (n19);
\draw (n0) edge[le3] (n21);
\draw (n0) edge[le3] (n31);
\draw (n1) edge[le3] (n2);
\draw (n1) edge[le3] (n3);
\draw (n1) edge[le1] (n7);
\draw (n1) edge[le3] (n13);
\draw (n1) edge[le1] (n17);
\draw (n1) edge[le3] (n19);
\draw (n1) edge[le1] (n21);
\draw (n1) edge[le3] (n30);
\draw (n2) edge[le1] (n3);
\draw (n2) edge[le1] (n7);
\draw (n2) edge[le3] (n8);
\draw (n2) edge[le3] (n9);
\draw (n2) edge[le3] (n13);
\draw (n2) edge[le3] (n27);
\draw (n2) edge[le3] (n28);
\draw (n2) edge[le3] (n32);
\draw (n3) edge[le3] (n7);
\draw (n3) edge[le1] (n12);
\draw (n3) edge[le1] (n13);
\draw (n4) edge[le1] (n6);
\draw (n4) edge[le1] (n10);
\draw (n5) edge[le3] (n6);
\draw (n5) edge[le3] (n10);
\draw (n5) edge[le1] (n16);
\draw (n6) edge[le1] (n16);
\draw (n8) edge[le2] (n30);
\draw (n8) edge[le3] (n32);
\draw (n8) edge[le2] (n33);
\draw (n9) edge[le2] (n33);
\draw (n13) edge[le3] (n33);
\draw (n14) edge[le2] (n32);
\draw (n14) edge[le3] (n33);
\draw (n15) edge[le2] (n32);
\draw (n15) edge[le3] (n33);
\draw (n18) edge[le3] (n32);
\draw (n18) edge[le2] (n33);
\draw (n19) edge[le3] (n33);
\draw (n20) edge[le2] (n32);
\draw (n20) edge[le3] (n33);
\draw (n22) edge[le2] (n32);
\draw (n22) edge[le3] (n33);
\draw (n23) edge[le3] (n25);
\draw (n23) edge[le2] (n27);
\draw (n23) edge[le2] (n29);
\draw (n23) edge[le3] (n32);
\draw (n23) edge[le2] (n33);
\draw (n24) edge[le2] (n25);
\draw (n24) edge[le2] (n27);
\draw (n24) edge[le3] (n31);
\draw (n25) edge[le2] (n31);
\draw (n26) edge[le2] (n29);
\draw (n26) edge[le3] (n33);
\draw (n27) edge[le3] (n33);
\draw (n28) edge[le2] (n31);
\draw (n28) edge[le3] (n33);
\draw (n29) edge[le3] (n32);
\draw (n29) edge[le3] (n33);
\draw (n30) edge[le3] (n32);
\draw (n30) edge[le3] (n33);
\draw (n31) edge[le3] (n32);
\draw (n31) edge[le3] (n33);
\draw (n32) edge[le2] (n33);
\end{pgfonlayer}
\end{tikzpicture}

\caption{
\label{figure:karate}
Strong edges in the Karate-club dataset inferred by the algorithm 
of \citet{sintos2014using} (left) and our method (right) using two teams.  
The colors of the edges and the vertices depict the two teams.}

\end{figure}

The growth of online social networks has been an important factor in 
shaping our lives for the 21st century. 
68\,\% of adults in the US, 
also accounting for those who do not use internet at all, 
are facebook users.\footnote{\url{http://www.pewinternet.org/2016/11/11/social-media-update-2016/}}
Over the past few years, an ecosystem of online social-network platforms has emerged, 
serving different needs and purposes: 
being connected with close friends, 
sharing news and being informed, 
sharing photos and videos, 
making professional connections, and so on.

The emergence of such social-networking platforms has introduced many novel research directions.
First, online systems have enabled recording and studying human behavior at a very large scale. 
Second, the specific features of the different systems
are changing the way people interact with each other: 
new social norms are formed and human behavior is adapting. 
Consequently, data collected by online social-network systems
are used to analyze and understand human behavior and complex social phenomena.
Questions of interest include
understanding information-diffusion phenomena, 
modeling network evolution and predicting future behavior, 
identifying the role of users and network links,
and more.

A question of particular importance, 
which is the focus of this paper, 
is the problem of inferring {\em the strength of social ties} in a network.
Quantifying the strength of social ties
is an essential task for sociologists interested in understanding complex network dynamics
based on pair-wise interactions~\cite{granovetter1973strength}, 
or for engineers interested in designing applications 
related to viral marketing~\cite{de2014facebook} 
or friend recommendation~\cite{lu2010link}.

The problem of inferring the strength of social ties in a network
has been studied extensively in the graph-mining community~\cite{gilbert2009predicting,gilbert2012predicting,onnela2007structure,sintos2014using,tang2012inferring,xiang2010modeling}.
While most approaches use user-level features in order to estimate the social-tie strength
between pairs of users, our approach, 
inspired by the work of Sintos and Tsaparas~\cite{sintos2014using}, 
relies on the {\em strong triadic closure} (\stc) principle~\cite{davis1970clustering,easley2010networks,granovetter1973strength}.
The \stc principle assumes that there are two types of ties in the social network: 
{\em strong} and {\em weak}. 
It then asserts that it is unlikely to encounter a triple of users
so that two of the ties are strong while the third is missing. 
In other words, two users who have a strong tie to a  third common friend
should be acquainted to each other, i.e., 
they should have {\em at least} a weak tie to each other.

Sintos and Tsaparas~\cite{sintos2014using}
address the problem of inferring the strength of social ties 
(i.e., labeling the links of a given network as {\em strong} or {\em weak})
by leveraging the \stc property in an elegant manner.
They first assume that users are more interested in 
establishing and maintaining strong ties, 
as presumably, this is the reason that they joined the network.
Using this assumption they
formulate the link-strength inference problem 
by asking to assign the maximal number of strong ties
(or the minimal number of weak ties)
so that the \stc property holds.
They prove that the problem is \NP-hard
and they devise an approximation algorithm for the 
variant of minimizing the number of weak ties.

In this paper, 
in addition to the network structure, 
we also consider as a collection of 
topical communities $C_1,\ldots, C_k$.
We assume that the given
communities are {\em tight},
that is, each community $C_i$ represents 
a set of users with {\em focused} interest at a particular topic.
For example, such a tight community may be
($i$) a set of users who have been actively involved
in a discussion in the social network about a certain issue,
($ii$) the set of scientists who work on `deep learning,'
or
($iii$) the HR team of a company.

We then require that each given community $C_i$ should be connected via strong ties.
In other words, for every two nodes in $C_i$ there is a path made of {\em strong} ties.
This requirement reflects the fact that we consider
tight communities, as the examples above.
Clearly this constraint is less meaningful if we consider {\em loose} communities,
i.e., all facebook users who like the `Friends' TV series.

Equipped with these assumptions we now define the 
problem of inferring the strength of social ties: 
given a social network $G=(V,E)$, 
and a set of tight communities $C_1,\ldots, C_k\subseteq V$, 
we ask to label all the edges in $E$ as either {\em strong} or {\em weak} so that
($i$)
each community $C_i$ is connected via strong ties; and 
($ii$)
the total number of \stc violations is minimized.
Our problem definition captures two natural phenomena:
first, tight communities tend to have a backbone, 
e.g., being part of a community implies a strong connection to one of the existing members. 
Second, strong ties tend to close triangles, 
as postulated by the strong triadic closure principle, 
and thus, 
real-world social networks have relatively few \stc violations.

\vspace{2mm}
\noindent
{\bf Example.}
An illustration of our method on the Karate-club 
dataset~\cite{zachary1977information} is shown in 
Figure~\ref{figure:karate}.
Our method (right) is contrasted with the 
algorithm of \citet{sintos2014using} (left).
Both approaches use the \stc principle, 
but additionally, our method requires that certain communities
provided as input are connected with strong ties.
In the example, we consider the two ground-truth communities
of the Karate-club dataset. 
We observe that the sets of strong ties inferred by the two methods
are fairly similar.
We also observe that our method introduces an \stc violation
only when it is necessary for ensuring connectivity. 
On the other hand, the method of \citet{sintos2014using} 
leaves several disconnected singleton nodes,  
which is less intuitive.
$\Box$
\vspace{2mm}

We capture the above intuition using two related problem definitions. 
For the first problem (\prbminstc) we ask to minimize the number of \stc violations, 
while for the second problem (\prbmaxtri) we ask to maximize the number of non-violated open triangles
--- there cannot be a violation on a closed triangle.
In both cases we label the network edges 
so as to satisfy the connectivity constraint, with respect to strong edges, 
for all input communities. 

We show that both problems, \prbminstc and  \prbmaxtri, 
are \NP-hard, even if the input consists of one community. 
Furthermore, we show that \prbminstc is hard to approximate to any multiplicative factor. 
On the other hand, the problem \prbmaxtri is amenable to approximation:
its objective function is submodular and non-decreasing, 
while the connectivity constraints can be viewed as an intersection of matroids. 
Thus, the classic result of \citet{fisher1978analysis} applies, 
implying that a greedy algorithm leads to $1/(k + 1)$ approximation ratio. 

We evaluate our methods on real-world networks and input communities. 
Our quantitative results show that our method achieves a balance
between baselines that optimize \stc violations and community connectivity separately, 
while our case study suggests the strong edges selected by the method are meaningful and intuitive.

The remaining paper is as follows.
We introduce the notation and give the problem definition in Section~\ref{sec:prel}.
We show the computational hardness in Section~\ref{sec:nphard} and present the approximation
algorithm in Section~\ref{sec:approx}. The related work is discussed in Section~\ref{sec:related},
and the experimental evaluation is given in Section~\ref{sec:exp}. We conclude the paper with
remarks in Section~\ref{sec:conclusions}.

\section{Preliminaries and problem definition}
\label{sec:prel}

The main input for our problem is an undirected graph $G = (V, E)$ with $n$
vertices and $m$ edges. 
Given a subset of vertices $X\subseteq V$ and a subset of edges $F\subseteq E$, 
we write $F(X)$ to denote the edges in $F$ that connect vertices in $X$.

We are interested in labeling the set of edges $E$. 
Specifically, we want to label each edge as either {\em strong} or {\em weak}.

To specify a labeling of edges $E$ it is sufficient
to specify the set of strong edges $S \subseteq E$. 
To quantify the quality of a labeling $S \subseteq E$, for a given graph $G = (V, E)$,  
we use the \emph{strong triadic closure} ({\stc}) property.
Namely, given a triple $(u, v, w)$ of vertices such that $(u, v), (v, w) \in S$, 
we say that the triple violates the {\stc} property if $(u, w) \notin S$. 
In other words, a strong friend of a strong friend must
be connected, possibly with a weak edge.
We define $\violations(S; G)$ to be the number of {\stc} violations. 
Typically, $G$ is known from the context, and we omit it from the notation.

let $S \subseteq E$ be the set of edges considered to be strong. 
Assume a triple $(u, v, w)$ of vertices such that $(u, v), (v, w) \in S$.
We say that the triple violates \emph{strong triadic closure} ({\stc}) property
if $(u, w) \notin S$. In other words, two users who have a common strong friend should
be connected, possibly with a weak tie.
We define $\violations(S; G)$ to be the number of {\stc} violations. 
Typically, $G$ is known from the context, and we 
simply denote the number of violations by $\violations(S)$.

As discussed in the introduction, our goal is to discover a strong backbone of the
graph. At simplest we are looking for a set of edges that connect the whole
graph with strong ties while minimizing the number of violations.

We also consider a more general case, where we are given a \emph{set of communities},
and the goal is to ensure that
each community is connected with strong ties.

More formally, we have the following problem definition

\begin{problem}[\prbminstc]
\label{problem:stc}
Given a graph $G=(V,E)$ and a set of communities $C_1,\ldots, C_k\subseteq V$,
find a set of strong edges $S\subseteq E$ such that
each
$(C_i, S(C_i))$ is connected 
and the number
of {\stc} violations, $\violations(S)$, is minimized.
\end{problem}

In the above problem definition $(C_i, S(C_i))$ is the subgraph of $G$
induced by the vertices in $C_i$ {\em and} the edges in $S$, 
that is, 
$S(C_i) = \{(u,v)\in E \mid u,v\in C_i \mbox{ and } (u,v)\in S\}$.

In order for \prbminstc to have at least one feasible solution, we assume that $(C_i,
E(C_i))$ is connected for each $C_i$.

In addition to minimization version, we consider a maximization version
of the problem.
In order to do that, given a graph $G$, let $T$ be the number of open triangles in $G$. 
We define $\covered(S) = T - \violations(S)$ to be the number of open triangles that are not
violated. 

This leads to the following optimization problem.
\begin{problem}[\prbmaxtri]
\label{problem:stc}
Given a graph $G=(V,E)$ and a set of communities $C_1,\ldots, C_k\subseteq V$,
find a set of strong edges $S$ such that
each
$(C_i, S(C_i))$ is connected,
and the number of non-violated triangles, $\covered(S)$, is maximized.
\end{problem}

Note that \stc violations can occur only for open triangles.
Therefore, $\covered(S)= T - \violations(S)$ is nonnegative, 
while it achieves its maximum value $T$ when there are no \stc violations.

Obviously, \prbminstc and \prbmaxtri have the same optimal answer.  However, we
will see that they yield different approximation results: \prbminstc cannot
have any multiplicative approximation guarantee (constant or non-constant) while a greedy algorithm has $1/(k + 1)$
guarantee for \prbmaxtri.

\section{Computational complexity}
\label{sec:nphard}

Our next step is to establish that \prbminstc (and \prbmaxtri) are \np-hard.
Moreover, we show that \prbminstc cannot have any multiplicative approximation
guarantee.

\begin{proposition}
Deciding whether there is a solution
\prbminstc with zero violations is \NP-complete. Thus, there is no multiplicative approximation
algorithm for \prbminstc,
unless \Poly=\NP. The result holds even if we use only one community.
\end{proposition}

\begin{proof}
To prove the result we will reduce \textsc{Clique\-Cover} to \prbminstc.
In an instance of \textsc{Clique\-Cover}, we are asked to partition a graph $G = (V, E)$
to $k$ subgraphs, each one of them being a clique.

Assume a graph $G = (V, E)$, where $V = \enset{v_1}{v_n}$, and an integer $k$.
We can safely assume that $G$ contains at least one singleton vertex, say $v_1$.
Otherwise, we can add a singleton vertex to $G$ and increase $k$ without changing
the outcome of \textsc{Clique\-Cover}.

For the reduction of  \textsc{Clique\-Cover} to \prbminstc, 
we first define a graph $H = (W, A)$.
The vertex set $W$ consists of $2n + k$ vertices grouped in 3 sets: the first set are the original vertices $V$, the
second set is $U$ with $n$ vertices, the third set is $X$
containing $k$ vertices.

The edges $A$ are as follows: We keep the original edges $E$. For each $i = 1,
\ldots, n$ and $j = 1, \ldots, k$, we add $(u_i, v_i)$, $(v_i, x_j)$, and $(u_i,
x_j)$.  We also fully-connect $X$.

We add one community consisting of the whole graph.

We claim that there is a 0-solution to \prbminstc if and only if there is a
clique cover for $G$. Since \textsc{Clique\-Cover} is \np-hard, this
automatically proves the inapproximability.

Assume first that we are given a clique cover $\mathcal{P}=\{P_1,\ldots,P_k\}$. 
Define the following set
of strong edges. For each vertex $v_i$, let $P_j$ be the clique containing $v_i$;
add edges $(u_i, v_i)$, $(v_i, x_j)$ to $S$. Finally, add an edge $(x_1, x_j)$ for each $j = 2, \ldots, k$.
It is straightforward to see that the connectivity constraints are satisfied.
The strong wedges are
\[
\begin{split}
u_i\text{--}v_i\text{--}x_j,&\quad\text{for}\quad v_i \in P_j, \\
v_i\text{--}x_j\text{--}x_1,&\quad\text{for}\quad v_i \in P_j, \text{ and } j \neq 1,\\
v_i\text{--}x_1\text{--}x_j,&\quad\text{for}\quad v_i \in P_1, \text{ and } j \neq 1, \\
x_j\text{--}x_1\text{--}x_q,&\quad\text{for}\quad q \neq j, \\
v_i\text{--}x_j\text{--}v_\ell,&\quad\text{for}\quad v_i, v_\ell \in P_j, \text{ and } i \neq \ell. \\
\end{split}
\]
None of these wedges induce a violation, the last one follows from the fact that $\mathcal{P}$ is a clique cover.
Thus $\violations(S) = 0$.

To prove the other direction,
let $S$ be the set of strong edges such that $\violations(S) = 0$.

Fix $i = 1, \ldots, n$.
To satisfy the connectivity, $(u_i, x_j) \in S$ or $(u_i, v_i) \in S$ (or both) for some $j$.
Define $Y = \set{v_i; (u_i, x_j) \in S}$ and $Z = V \setminus Y$.

Let $v_i \in Z$. Since $(u_i, v_i) \in S$, all edges adjacent to $v_i$ in $E$ are weak.
Thus, to satisfy the connectivity, we must have $(v_i, x_j)$ for some $j$.

Define two families $\mathcal{A}$ and $\mathcal{B}$, each of $k$ sets, by
\[
\begin{split}
	A_j & = \set{v_i \in Y; (u_i, x_j) \in S}
	\quad\text{and}\quad  \\
	B_j & = \set{v_i \in Z; (v_i, x_j) \in S}.
\end{split}
\]
Write $A_0 = B_0 = \emptyset$, and define a family $\mathcal{P}$ of $k$ disjoint sets
by $P_j = (A_j \cup B_j) \setminus P_{j - 1}$. $\mathcal{P}$ covers $V$ since
each vertex in $V$ is in $A_j$ or $B_j$ for some $j$.

We claim that $\mathcal{P}$ is a clique cover. To see this, let $v_i, v_\ell \in P_j$.
If $v_i \in Y$ and $v_\ell \in Z$, then $u_i$--$x_j$--$v_\ell$ is a violation since $i \neq \ell$.
If $v_i, v_\ell \in Y$, then $u_i$--$x_j$--$u_\ell$ is a violation, or $i = \ell$. 
If $v_i, v_\ell \in Z$, then either $i = \ell$ or $(v_i, v_\ell) \in E$.
This shows that $\mathcal{P}$ is a clique cover. 
\end{proof}

\begin{corollary}
The \prbmaxtri problem is \NP-hard. 
The result holds even if we use only one community.
\end{corollary}

\section{Approximation algorithm}\label{sec:approx}

In the previous section we saw that the problems
\prbminstc and \prbmaxtri are \NP-hard, even for one community,
and additionally, \prbminstc is hard to approximate to any multiplicative factor. 
In this section we show that \prbmaxtri can be approximated with 
$1 / (k + 1)$ guarantee, 
where $k$ is the number of communities in the input.  
As an imporant consequence, if we have one community, 
we can find a solution with approximation guarantee $1 / 2$.
Furthermore, it follows that if all communities are edge-disjoint,
our algorithm yields a $1/2$ approximation guarantee.

To prove the approximation algorithm we argue that $\covered(\cdot)$ is submodular
with respect to weak edges. 
Moreover, the connectivity constraint of
each communitiy can be viewed as a matroid.
Thus, satisfying all the connectivity constraints 
is an intersection of matroids.

These properties 
allow us to use a classic result of maximizing a submodular
function over an intersection of $k$ matroids:
Fisher et al.~\cite{fisher1978analysis} showed that a greedy algorithm
leads to $1/(k + 1)$ approximation ratio. 
Here the greedy algorithms starts
with none of the edges being weak, that is, all edges are strong.  
We find a strong edge, say $e$, inducing the most violations.  
We convert $e$ to a weak edge if the connectivity constraints allow it. 
Otherwise, we let $e$ being strong. 
The pseudo-code is given in Algorithm~\ref{alg:greedySConnect}.

Note that our problem formulation is agnostic
with respect to whether strong edges should be maximized or minimized.
In Algorithm~\ref{alg:greedySConnect}
strong edges are kept, even if they are not crucial for connectivity, 
as long as they do not induce any violations. 
This behavior is in line with the idea of \citet{sintos2014using},
who aim to maximize the number of strong edges. 
It is in contrast, however, with our second baseline, 
the algorithm of \citet{angluin13connectivity}, 
who want to find a minimum set of edges to ensure connectivity. 
If we wish to obtain a minimal number of strong edges, 
we can continue the main iteration in Algorithm~\ref{alg:greedySConnect}
and convert to weak all edges that are not necessary for connectivity
and do not create any \stc violations.

\begin{algorithm}[t]
	\caption{Greedy algorithm for \prbmaxtri}
	\label{alg:greedySConnect}
	$S \define E$; 
	$A \define E$\;	
	\While{$A\ne \emptyset$} {
		$e=\argmax_{e \in A}\covered(S \setminus \set{e})$\; 
		\If{$S\setminus \{e\}$ satisfies the connectivity constraints} { 
			$S \define S\setminus\{e\}$\;
		}
		$A \define A\setminus\{e\}$\;
	}
	\Return $S$\;
\end{algorithm}

We now show the properties required by the result of Fisher et al.~\cite{fisher1978analysis}
for the greedy algorithm to yield approximation ratio $1/(k + 1)$. 
We first show that the function $\covered(\cdot)$ is submodular with respect to weak edges.

\begin{proposition}
Consider a graph $G = (V, E)$.
Let $f(W) = \covered(E \setminus W)$.
The function $f$ is submodular and non-decreasing.
\end{proposition}

\begin{proof}
To prove the submodularity we show that $\violations(\cdot)$ is supermodular
with respect to strong edges. 
This makes $\covered(\cdot)$ submodular with respect to strong edges, 
which in turn makes $f$ submodular with respect to weak edges.
For the last implication it is well-known
that a function is submodular if and only if its complement is
submodular.\!\footnote{see, for example,
\url{http://melodi.ee.washington.edu/~bilmes/ee595a_spring_2011/lecture1_presented.pdf}
for a proof.}

Let $S$ be a set of strong edges. For $a$ vertex $u$, define $N_{X}(u) = \set{
v \mid (u,v)\in X}$ to be the strong-neighbors of $u$. Define also $\overline{N}(u) =
\{ v \mid (u,v)\not\in E\}$ to be the non-neighbors or $u$.

The number of additional \stc violations introduced by labeling edge $e$ as
strong is 
\[ 
\violations(S\cup\{ e\}) - \violations(S)
 = \abs{ N_{S}(u) \cap \overline{N}(v) } 
 + \abs{ N_{S}(v) \cap \overline{N}(u) }.
\]
Let $T \subseteq S$.
For any $u \in V$ and any edge set $W$, we have 
\begin{equation}
\label{eq:monotone}
	\abs{ N_{T}(u) \cap W } \le \abs{ N_{S}(u) \cap W }. 
\end{equation}
This implies that for $T\subseteq S\subseteq E$
and any edge $e \notin S$, 
\[ 
\violations(T\cup\{ e\}) - \violations(T)
\le \violations(S\cup\{ e\}) - \violations(S),
\]
which proves the supermodularity of $\violations(\cdot)$.

To prove the monotonicity, we show that $\violations(\cdot)$ is non-decreasing.
This makes $\covered(\cdot)$ non-increasing, which makes $f$ non-decreasing.

Let $T \subseteq S \subseteq E$.
The number of violations induced by set $S$ is half of the sum of violations
induced by each edge $(u,v)\in S$,
\[
\violations(S) =  
\frac{1}{2}\sum_{(u,v)\in S}
\abs{ N_{S}(u) \cap \overline{N}(v) } 
	+ \abs{ N_{S}(v) \cap \overline{N}(u) }.
\]
One half is needed as each violated triangle
is considered twice, because it is caused by two strong edges.
Similarly,
\[
\violations(T) =  
\frac{1}{2}\sum_{(u,v)\in T}
\abs{ N_{T}(u) \cap \overline{N}(v) } 
	+ \abs{ N_{T}(v) \cap \overline{N}(u) }.
\]

Since $T \subseteq S$, each edge occuring in the sum
occurs also in the sum for $\violations(S)$. 
Moreover, Equation~(\ref{eq:monotone}) guarantees 
that the term corresponding to an edge $(u, v)$ in $\violations(T)$ is smaller than the
than the term corresponding to the same edge $(u, v)$ in $\violations(S)$.
Consequently, $\violations(\cdot)$ is non-decreasing.
\end{proof}

Our next step is to argue that the connectivity constraints are matroids with
respect to weak edges.  
Fortunately, this is a known result and these matroids are commonly known as \emph{bond matroids},
see for example, Proposition 3.3~by~\citet{oxley2003matroid}.

\begin{proposition}
\label{prop:bond}
Assume a graph $G = (V, E)$ and a subset $C$ such that $(C, E(C))$ is connected.
Define a family of sets
\[
	\mathcal{M} = \set{W \subseteq E ; (C, E(C) \setminus W(C)) \text{ is connected}}.
\]
Then $\mathcal{M}$ is matroid.
\end{proposition}

The two propositions show that we can use the result by~\citet{fisher1978analysis}, and obtain
$1 / (k + 1)$ guarantee, where $k$ is the number of communities, 
i.e., the sets of vertices for which we require a connectivity constraint.

We can obtain a better guarantee, $1/2$, if we know that communities are edge-disjoint.
This follows from the fact that we can express the connectivity constraints as a single matroid.

\begin{proposition}
Assume a graph $G = (V, E)$ and family of edge-disjoint subsets $C_1, \ldots, C_k$,
such that each $(C_i, E(C_i))$ is connected.
Define a family of sets
\[
	\mathcal{M} = \set{W \subseteq E ; (C_i, E(C_i) \setminus W(C_i)) \text{ is connected, for every }i}.
\]
Then $\mathcal{M}$ is matroid.
\end{proposition}

The result follows from immediately
Proposition~\ref{prop:bond}.
and the following standard lemma which we state without a proof.

\begin{lemma}
Let $\mathcal{M}_1, \ldots, \mathcal{M}_k$ be $k$ matroids, each
matroid $\mathcal{M}_i$ is defined over its own ground set $U_i$.
Then a direct sum
\[
	\mathcal{M} = \set{\bigcup_{i = 1}^k X_i \mid X_i \in \mathcal{M}_i}
\]
is a matroid over $\bigcup_{i = 1}^k U_i$.
\end{lemma}

Thus we can use the result by~\citet{fisher1978analysis} but now we have only one matroid instead
of $k$ matroids. This gives us an approximation guarantee of $1/2$.

\spara{Computational complexity}:
Let us finish with the compu\-ta\-tional-complexity analysis of the greedy algorithm.
Assume that given a graph $G = (V, E)$, 
we have already enumerated all open triangles. 
Let $t$ be the number of such triangles.

During the while-loop of the greedy algorithm, we maintain a priority queue for
$m$ edges, prioritized with the number of violations induced by a single edge.
Whenever, a strong edge is deleted, we visit every open triangle induced by
this edge and reduce the number of violations of the strong sister edge by 1.
Note that we visit every triangle at most twice, so maintaining the queue
requires $\bigO(t + m \log n)$ time, if we use Fibonacci heap. To check the
connectivity, we can use the technique introduced by~\citet{holm2001poly}, allowing us to
do a connectivity check in $\bigO(\log^2n)$ amortized time.  
Thus, in total we need $\bigO(t + k m \log^2 n )$
time, plus the time to build the list of open triangles. 
Building such a list can be done in
$\bigO\pr{\sum_v \deg^2(v)}$ time.

\section{Related work}
\label{sec:related}

The study of interpresonal ties has a long history in social psychology.
Several researchers have investigated the role of different types of social ties 
with respect to structural properties of social networks, 
as well as with respect to information-propagation phenomena. 
For example, in economics, 
\citet{montgomery1992job}
showed that weak ties are positively correlated 
to higher wages and higher aggregate employment rates.
More recent works 
considered how different social ties are formed and how they evolve
in online social networks, such as 
email networks~\cite{kossinets2006empirical} and
mobile-phone networks~\cite{onnela2007structure}.

The strong triadic closure (\stc) property, 
which forms the basis of our inference algorithm, 
was first formulated in the seminal paper of 
\citet{granovetter1973strength}, 
while evidence that this property holds in social networks 
has appeared in earlier works~\cite{davis1970clustering,newcomb2961acquaintance}. 
\citet{memic2009testing} have conducted a more recent study confirming 
that the principle remains valid on more
recently-collected datasets~\cite{memic2009testing}.

In computer science, 
there have been several works that study the problem of 
inferring the strength of social ties in a network.

\citet{kahanda2009using} 
use transactional events,  
such as communication and file transfers, 
to predict link strength, 
by applying techniques from the literature of the 
link-prediction problem~\cite{liben2007link}.
It is shown that the approach can accurately predict strong relationships.
\citet{gilbert2009predicting} propose a predictive model
for inferring tie strength. 
The model uses variables describing the interaction of users in a social-media platform. 
The paper also illustrates how the inferred tie strength can be used to improve 
social-media features, such as, privacy controls, message routing, and friend recommendation.
Likewise, \citet{xiang2010modeling} leverage user-level interaction data. 
They formulate the problem of inferring hidden relationship strengths
using a latent-variable model, which is learned by a coordinate-ascent optimization procedure.
A feature of their setting is that social strengths are modeled as real-valued variables, 
not just binary. 
\citet{jones2013inferring} examined a large set of features for the task of
predicting the strength of social ties
on a Facebook interaction dataset, 
and found that the frequency of online interaction is diagnostic of strong ties,
while
private communications (messages) are not necessarily more informative than 
public communications (comments, wall posts, and other interactions).

\citet{backstrom2014romantic}
consider a particular type of social ties --- romantic relationships ---
and they ask whether this can be accurately recognized. 
They use a large sample of Facebook data to answer the question affirmatively, 
and on the way they develop a new type of tie strength,
{\em the extent to which two people's mutual friends are not themselves well-connected},
which they call ``dispersion.''

In a different direction,
\citet{fang2015uncovering} consider only closed triangles
and ask whether it is possible to find out which edges are formed last, 
i.e., which edges closed an open triad.
The underlying research question is to recover the 
dynamic information in the triadic-closure process. 
They approach this problem using a probabilistic factor-graph model, 
and apply the proposed model on a large collaboration network.

Researchers have also studied the tie-strength inference problem
in the presence of more than one social network.
\citet{gilbert2012predicting} explore how well a tie strength model developed for one social-media 
platform adapts to another, 
while 
\citet{tang2012inferring}
consider a generalization of the problem over multiple heterogeneous networks.
Their work uses a transfer-based factor-graph model, 
and also incorporates features motivated from social-psychology theories, 
such as 
social balance \cite{easley2010networks},
structural holes \cite{burt2009structural}, and
and social status \cite{davis1972structure}.

Most of the above works on tie-strength inference 
utilize pairwise user-level interaction data, 
such as email, private messages, public mentions, frequency of interactions, and so on. 
In many cases such detailed data are not available. 
Our objective is to address the tie-strength inference problem 
using non private data, such as the structure of the social network 
and information about communities and teams that users have participated.

Conceptually and methodologically our paper is related to the 
work of \citet{sintos2014using}, 
who use as available information only the network structure, 
to infer strong and weak ties with the means of the 
strong triadic closure property~\cite{easley2010networks}.
We extend that work by introducing community-level information
and a corresponding connectivity constraint
to account for explaining the observed community structure:  
namely, we require that each community should be connected via strong ties. 
Like the work of \citet{sintos2014using},
we follow a combinatorial approach, 
but the techniques we use are significantly different.

A problem related to the inference of tie strength
is the problem of predicting edge signs in 
social networks~\cite{chiang2014prediction,leskovec2010predicting}. 
The sign of an edge is typically interpreted as `friend' or `foe', 
and thus, existing algorithms utilize theories from social psychology
that are developed for this kind of relationships, 
in particular social balance~\cite{easley2010networks} and 
social status theory~\cite{davis1972structure}.

In our experiments we are comparing our method with the algorithm of 
\citet{angluin13connectivity}, 
which takes as input a set of teams (communities) over a set of entities
and seeks to add a minimal number of edges among the entities
so that all given teams are connected. 
The algorithm is greedy and it is shown to have a $\bigO(\log n)$ 
approximation guarantee.
In our case, in addition to the set of teams
we also have as input an underlying network, 
and edges are selected only if they are network edges. 
Selected edges are considered strong, 
and non-selected edges are considered weak.
Thus, the method of \citet{angluin13connectivity}
is a combinatorial approach, 
which aim to satisfy connectivity among the input communities (like our method), 
but it does not take into account \stc violations.
On the other hand, it aims to minimize the number of strong edges
(while our method is oblivious to this consideration).

\section{Experimental evaluation}\label{sec:exp} %
\begin{table*}[t]
		\caption{Network characteristics. $|V|$: number of vertices; $|\edges|$:
		number of edges in the underlying network; $|V_0|$: number of vertices, which participate in any given set (community); $|\indedges|$: number of edges induced by communities;
		$\ell$: number of sets (communities); 
		$\avg(\alpha_0)$: average density of subgraphs induced by input communities; 
		$s_{\min}$, $s_{\avg}$: minimum and average set size; 
		$t_{\max}$, $t_{\avg}$: maximum and average participation of a vertex to a set.}
		\label{tab:stats}
	\setlength{\tabcolsep}{0pt}
	\centering
	\begin{tabular*}{\textwidth}{@{\extracolsep{\fill}}l l rrrrrrrrr}
		\toprule 
		Dataset & $|V|$ &$|\edges|$&  $|V_0|$ &$|\indedges|$ & $\ell$&$\avg(\alpha_0)$ & $s_{\min}$&$s_{\avg}$& $t_{\max}$&$t_{\avg}$\\		
		\midrule
		{\DBLP}  &  10001& 27687& 10001 & 22264  & 1767 & 0.58 & 6 & 7.46 & 10 & 1.31\\
		{\youtube}  & 10002 & 72215 & 10001 & 15445 & 5323 & 0.69 & 2 & 4.02 & 82 & 2.14\\		
		{\KDD}  & 2891 & 11208 & 1598  &3322 & 5601  &0.96 & 2 & 2.40 & 107 & 8.41\\
		{\ICDM}  & 3140  &10689 & 1720 & 3135  &5937 & 0.96 & 2 & 2.34 & 139 & 8.11\\
		{\FBcirc} & 4039 & 88234 & 2888 & 55896 & 191 & 0.64 & 2 & 23.15 & 44 & 1.53\\
		{\FBfeat} & 4039 &  88234 &  2261  & 20522  & 1239 &  0.93 & 2 & 3.75 & 13 & 2.05\\
		{\lastFMart} & 1892 & 12717 & 1018 & 2323 & 2820 & 0.89 & 2 & 2.91 & 221 &  8.08\\
		{\lastFMtag}  & 1892 & 12717 & 855 & 1800 & 651 & 0.88& 2 & 3.43& 20 & 2.61\\
		{\DBbook}  & 1861& 7664& 932 &1145& 1288 &0.97& 2 & 2.27 & 27 & 3.13\\
		{\DBtag} & 1861&  7664 & 1507 & 2752 & 4167&  0.96 & 2 & 2.26 & 68 & 6.25\\	
		\bottomrule
	\end{tabular*}

\end{table*}

In this section we present our experimental evaluation.
We describe the datasets used in the evaluation, 
we discuss the baselines, 
and then present the evaluation results, 
including quantitative experiments and a case study.

The datasets and the implementation of the methods
used in our experimental evaluation are publicly available.\footnote{https://github.com/polinapolina/connected-strong-triadic-closure}

\spara{Datasets.}
We use $10$ datasets, each dataset
consists of a network and a set of communities. 
We describe these datasets below, 
while their basic characteristics are shown in 
Table~\ref{tab:stats}. To ensure connectivity of each community, we selected only one part of each disconnected community, which induces the largest connected component.

\smallskip
\noindent
$\bullet$ {\KDD} and {\ICDM} are subgraphs of the DBLP co-authorship network, 
restricted to articles published in the respective conferences.
Edges represent co-authorships between authors.
Communities are formed by keywords that appear in paper abstracts.

\smallskip
\noindent
$\bullet$ \FBcirc and \FBfeat are Facebook ego-networks available 
at the SNAP repository~\cite{snapnets}. 
In {\FBcirc} the communities are social-circles of users.
In {\FBfeat} communities are formed by user profile features.

\smallskip
\noindent
$\bullet$ 
{\lastFMart} and {\lastFMtag} are friendship networks of last.fm 
users.\footnote{\url{grouplens.org/datasets/hetrec-2011/}}
A community in {\lastFMart} and {\lastFMtag} is formed by users who listen to the
same artist and genre, respectively.

\smallskip
\noindent
$\bullet$ 
{\DBbook} and {\DBtag} 
are friendship networks of Delicious users.\footnote{\url{www.delicious.com}}
A community in \DBbook and \DBtag is formed by users who use the same 
bookmark and keyword, respectively.

\smallskip
Additionally, 
we use
SNAP datasets~\cite{snapnets} with ground-truth communities. 
To have more focused groups, 
we only keep communities with size less than~10.
To avoid having disjoint communities, 
we start from a small number of seed communities and iteratively add other communities
that intersect at least one of the already selected.
We stop when the number of vertices reaches $10\,000$.
In this way we construct the following datasets:

\smallskip
\noindent
$\bullet$ \DBLP:
This is also a co-authorship network.
Communities are defined by publication venues.

\smallskip
\noindent
$\bullet$ \youtube: 
This is a social network of Youtube users.
Communities consist of user groups created by users.

\spara{Baselines.}
There are no direct baselines to our approach since the problem definition is
novel. Instead we focus on comparing our method with the two techniques that
inspired our approach. The first method is by~\citet{sintos2014using} and it
maximizes the number of strong edges while keeping the number of \stc violations equal to $0$.
The second method is by~\citet{angluin13connectivity} and it minimizes the number
of edges needed to connect the communities.
We refer to these methods as \BLA and \BLS, respectively, and we call our method~\our.

Note that \BLA is oblivious to the \stc property while \BLS does not use any
community information.  As we combine both goals, we expect that \our
results in a compromise of these two baselines.

\begin{table*}[t]
	\caption{Characteristics of edges selected as strong by \our and the two baselines. 
	$b$: number of violated triangles in the solution divided by the number of open triangles (all possible violations); 
	$s$: number of strong edges in the solution divided by the number of all edges; 
	$c$: average number of connected components per community. 
	$\BLAalt$ corresponds to \BLA; $\BLSalt$ corresponds to \BLS.}
	\label{tab:results}
	\centering
	\begin{tabular*}{\textwidth}{@{\extracolsep{\fill}}l r lrr  r rrr r  rrr}
		\toprule 
		&& \multicolumn{3}{c}{\our} && \multicolumn{3}{c}{\BLA}  && \multicolumn{3}{c}{\BLS}\\
		\cmidrule{3-5}
		\cmidrule{7-9}
		\cmidrule{11-13}
		Dataset && $b$& $s$ & $c$&
		&$b_{\BLAalt}/b$& $s_{\BLAalt}/s$& $c_{\BLAalt}$&
		&$b_{\BLSalt}/b$& $s_{\BLSalt}/s$ &$c_{\BLSalt}$\\		
		\midrule
		{\DBLP} && 0.07 & 0.47& 1 &&2.77 & 0.77 & 1& &0.0 & 1.08&3.53\\
		{\youtube}   && 0.01& 0.16 &1& & 1.21&0.98&1&&0.0 &0.49&3.30\\
		{\KDD} && 0.08 &0.35&1& &1.09& 0.63& 1&&0.0&0.81&1.93\\
		{\ICDM}   && 0.07& 0.38 &1 &&1.06& 0.57&1&& 0.0& 0.83&1.84\\
		{\FBcirc} && 0.002& 0.15 & 1 &&61.05& 0.20&1&& 0.0&1.05&8.76\\ 
		{\FBfeat}&& 0.003& 0.12 &1 &&0.36& 0.22&1&& 0.0&1.35&2.41\\ 
		{\lastFMart} && 0.02&0.15 & 1 &&1.11& 0.78 &1&&0.0 &0.67&2.58\\ 
		{\lastFMtag}  &&0.008 & 0.12  &1& & 1.17& 0.68& 1&&0.0& 0.83&2.98\\ 
		{\DBbook}  &&0.01& 0.35 & 1 &&1.01& 0.35&1&&0.0& 1.04&1.61\\ 
		{\DBtag} && 0.10& 0.45&1 &&1.02& 0.66&1& &0.0& 0.80&1.74\\ 
		\bottomrule
	\end{tabular*}
\end{table*}

\spara{Comparison with the baselines.}
We compare the performance of \our with the two baselines.
We run all algorithms on our datasets
and measure the number of edges selected as strong, 
the number of \stc violations,
and the number of connected components created by strong edges for each
of the input communities
(so as to test the fragmentation of the communities).
The results are shown in Table~\ref{tab:results}.
The number of \stc violations and the number of strong edges 
are reported as ratios (see table) for easy comparison.
As expected, \BLA introduces more \stc violations than \our: 
typically between 1\%--21\%. 
Interestingly, \BLA introduces less violations in \lastFMtag and 60
times more violations in \lastFMart.

On the other hand, \BLS results in disconnected communities, 
ranging from $1.74$ to $8.76$ connected components per community, on average.

\begin{table}[t]
	\caption{Precision and recall of \BLA.}
	\label{tab:accuracyBLA}
	\centering
	\begin{tabular}{ l  r r r r }
		\toprule 
		Dataset & $P_W$ &  $R_W$ &  $P_S$ &  $R_S$ \\		
		\midrule
		{\KDD} &0.86&0.92&0.63&0.48\\ 
		{\ICDM} &0.87&0.93&0.66&0.50 \\ 
		{\lastFMart} &0.91&0.95&0.54&0.37\\ 
		{\lastFMtag} &0.92&0.95&0.26&0.16\\ 
		{\DBbook} &0.92&0.94&0.36&0.27\\ 
		{\DBtag} &0.82&0.87&0.50&0.41\\ 
		\bottomrule
	\end{tabular}
\end{table}

\begin{table}[t]
	\caption{Precision and recall of \BLS.}
	\label{tab:accuracyBLS}
	\centering
	\begin{tabular}{ l  r r r r  }
		\toprule 
		Dataset & $P_W$ &  $R_W$ &  $P_S$ &  $R_S$ \\		
		\midrule
		{\KDD} &0.78&0.70&0.19&0.26\\ 
		{\ICDM} & 0.77&0.66&0.18&0.28\\ 
		{\lastFMart} & 0.88&0.90&0.14&0.12\\ 
		{\lastFMtag} &0.91&0.89&0.09&0.11\\ 
		{\DBbook} &0.92&0.64&0.13&0.49\\ 
		{\DBtag} &0.75&0.62&0.22&0.35\\ 
		\bottomrule
	\end{tabular}
\end{table}

\begin{table}[t]
	\caption{Precision and recall of \our.}
	\label{tab:accuracy}
	\centering
	\begin{tabular}{ l  r r r r }
		\toprule 
		Dataset & $P_W$ &  $R_W$ &  $P_S$ &  $R_S$ \\		
		\midrule
		{\KDD} &0.85&0.75&0.36&0.51 \\ 
		{\ICDM} & 0.85&0.71&0.34&0.55 \\ 
		{\lastFMart} & 0.91&0.90&0.36&0.39 \\ 
		{\lastFMtag} &0.92&0.90&0.15&0.20 \\ 
		{\DBbook} &0.93&0.67&0.15&0.57 \\ 
		{\DBtag} &0.81&0.66&0.32&0.55 \\ 
		\bottomrule
	\end{tabular}
\end{table}

\smallskip
In the second experiment we 
test whether strong and weak ties can predict intra- and inter-community edges, respectively.
The rationale of this experiment
is to test the hypothesis that weak ties are bridges between different communities.
With respect to the different methods, 
our objective is to further demonstrate that \our method results 
as a middle ground between \BLA and \BLS.

We randomly select half of the communities as {\em test communities} 
and run \our and \BLA using as input the underlying network and the other half of the communities. 
We also run  \BLA using as input the underlying network; 
recall that this algorithm does not use any community information as input. 
Next, using the test communities we construct a set of intra-community edges
$E_{\mathit{intra}}$, consisting of edges that belong to at least one community, and 
inter-community edges $E_{\mathit{inter}}$, 
consisting of edges that bridge two communities (but do not belong to any single community).

Let us denote all strong edges in
the output of a given method as $S$ and the weak edges as $W$.
We define precision $P_W$ and recall $R_W$ for weak edges as 
\[
	P_W=\frac{|W\cap E_{\mathit{inter}}|}{|W|} \quad\text{and}\quad R_W=\frac{|W\cap E_{\mathit{inter}}|}{|E_{\mathit{inter}}|},
\]
and precision $P_S$ and recall $R_S$ for strong edges as 
\[
	P_S=\frac{|S\cap E_{\mathit{intra}}|}{|S|} \quad\text{and}\quad R_S=\frac{|S\cap E_{\mathit{intra}}|}{|E_{\mathit{intra}}|}. 
\]

\smallskip
\BLA selects greedily edges that connect as many communities as possible.
In other words, it prefers edges that are in many communities in the training set, 
and this acts as a strong signal for an edge being also in a community in the test set. 
The results
shown in Tables~\ref{tab:accuracyBLA}--\ref{tab:accuracy} support this intuition,
showing that \BLA obtains the best results. We also see that \BLS, which does not
use any community information, has the worst results, while our method is able
to improve \BLS by incorporating information from communities. 

\spara{Running time.}
Our implementation was done in Python and the bottleneck of the algorithm
is constructing the list of wedges. The running times vary greatly from dataset
to dataset. At fastest we needed 52 seconds while the at slowest we needed almost 5 hours.
We should point out that a more efficient implementation as well using parallelization
with constructing the wedges should lead to significant reduction in computational time.

\spara{Case study.}
To demonstrate a simple use case, we use a snippet of \KDD dataset: 
We picked five recent winners of SIGKDD innovation award: Philip S. Yu,
Hans-Peter Kriegel, Pedro Domingos, Jon M. Kleinberg and Vipin Kumar and
constructed an underlying network  as a union of their ego-nets.
We then used 5 common topics, \emph{cluster}, \emph{classif}, \emph{pattern},
\emph{network}, and \emph{distribut} as communities.
Figure~\ref{fig:dblp} depicts the discovered edges of \our. From the figure we see that
showing only strong edges significantly simplifies the graph. The selected strong edges
are reasonable: for example a path from Hans-Peter Kriegel to Pedro Domingo was
Arthur Zimek, Karsten Borgwardt, and Luc De Raedt, while a path from Pedro Domingo to Jon Kleinberg
was Luc De Raedt, Xifeng Yan, Zhen Wen, Ching-Yung Lin, Hang-hang Tong, Spiros Papadimitriou
Christos Faloutsos, and Jure Leskovec.

A zoom-in version of the graph of Figure~\ref{fig:dblp},
showing the names of all authors,
is omitted due to space constraints
but can be found in the public code and dataset 
repository.\footnote{https://github.com/polinapolina/connected-strong-triadic-closure}

\begin{figure*}
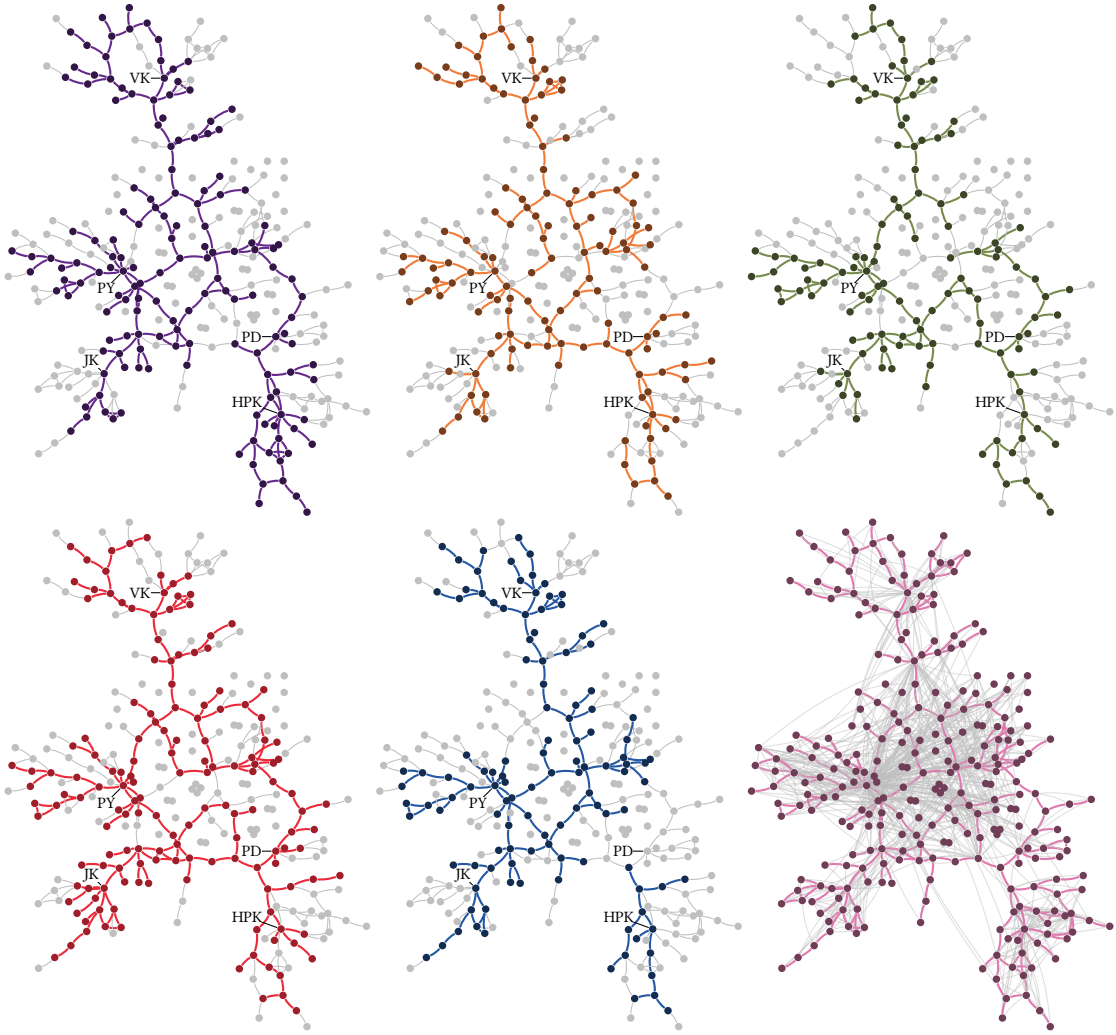

\begin{tikzpicture}[scale = 0.4]
\tikzstyle{lm1} = [circle, fill = yafcolor1!50!black, inner sep = 1pt]
\tikzstyle{lm2} = [circle, fill = gray!50, inner sep = 1pt]
\tikzstyle{ll} = []
\tikzstyle{le1} = [yafcolor1, thick, bend left = 10]
\tikzstyle{le2} = [gray!50, bend left = 10]
\input{dblp_graph/dblp_team1}
\draw (n140) -- ++(-135:0.8cm) node[fill=white, inner sep = 0pt] {\scriptsize PY};
\draw (n47) -- ++(135:0.6cm) node[fill=white, inner sep = 0pt] {\scriptsize JK};
\draw (n128) -- ++(180:0.8cm) node[fill=white, inner sep = 0pt] {\scriptsize PD};
\draw (n165) -- ++(180:0.8cm) node[fill=white, inner sep = 0pt] {\scriptsize VK};
\draw (n237) -- ++(160:1.2cm) node[fill=white, inner sep = 0pt] {\scriptsize HPK};
\end{tikzpicture}
\begin{tikzpicture}[scale = 0.4]
\tikzstyle{lm1} = [circle, fill = yafcolor2!50!black, inner sep = 1pt]
\tikzstyle{lm2} = [circle, fill = gray!50, inner sep = 1pt]
\tikzstyle{ll} = []
\tikzstyle{le1} = [yafcolor2, thick, bend left = 10]
\tikzstyle{le2} = [gray!50, bend left = 10]
\input{dblp_graph/dblp_team2}
\draw (n140) -- ++(-135:0.8cm) node[fill=white, inner sep = 0pt] {\scriptsize PY};
\draw (n47) -- ++(135:0.6cm) node[fill=white, inner sep = 0pt] {\scriptsize JK};
\draw (n128) -- ++(180:0.8cm) node[fill=white, inner sep = 0pt] {\scriptsize PD};
\draw (n165) -- ++(180:0.8cm) node[fill=white, inner sep = 0pt] {\scriptsize VK};
\draw (n237) -- ++(160:1.2cm) node[fill=white, inner sep = 0pt] {\scriptsize HPK};
\end{tikzpicture}
\begin{tikzpicture}[scale = 0.4]
\tikzstyle{lm1} = [circle, fill = yafcolor3!50!black, inner sep = 1pt]
\tikzstyle{lm2} = [circle, fill = gray!50, inner sep = 1pt]
\tikzstyle{ll} = []
\tikzstyle{le1} = [yafcolor3, thick, bend left = 10]
\tikzstyle{le2} = [gray!50, bend left = 10]
\input{dblp_graph/dblp_team3}
\draw (n140) -- ++(-135:0.8cm) node[fill=white, inner sep = 0pt] {\scriptsize PY};
\draw (n47) -- ++(135:0.6cm) node[fill=white, inner sep = 0pt] {\scriptsize JK};
\draw (n128) -- ++(180:0.8cm) node[fill=white, inner sep = 0pt] {\scriptsize PD};
\draw (n165) -- ++(180:0.8cm) node[fill=white, inner sep = 0pt] {\scriptsize VK};
\draw (n237) -- ++(160:1.2cm) node[fill=white, inner sep = 0pt] {\scriptsize HPK};
\end{tikzpicture}

\begin{tikzpicture}[scale = 0.4]
\tikzstyle{lm1} = [circle, fill = yafcolor7, inner sep = 1pt]
\tikzstyle{lm2} = [circle, fill = gray!50, inner sep = 1pt]
\tikzstyle{ll} = []
\tikzstyle{le1} = [yafcolor4, thick, bend left = 10]
\tikzstyle{le2} = [gray!50, bend left = 10]
\input{dblp_graph/dblp_team4}
\draw (n140) -- ++(-135:0.8cm) node[fill=white, inner sep = 0pt] {\scriptsize PY};
\draw (n47) -- ++(135:0.6cm) node[fill=white, inner sep = 0pt] {\scriptsize JK};
\draw (n128) -- ++(180:0.8cm) node[fill=white, inner sep = 0pt] {\scriptsize PD};
\draw (n165) -- ++(180:0.8cm) node[fill=white, inner sep = 0pt] {\scriptsize VK};
\draw (n237) -- ++(160:1.2cm) node[fill=white, inner sep = 0pt] {\scriptsize HPK};
\end{tikzpicture}
\begin{tikzpicture}[scale = 0.4]
\tikzstyle{lm1} = [circle, fill = yafcolor5!50!black, inner sep = 1pt]
\tikzstyle{lm2} = [circle, fill = gray!50, inner sep = 1pt]
\tikzstyle{ll} = []
\tikzstyle{le1} = [yafcolor5, thick, bend left = 10]
\tikzstyle{le2} = [gray!50, bend left = 10]
\input{dblp_graph/dblp_team5}

\draw (n140) -- ++(-135:0.8cm) node[fill=white, inner sep = 0pt] {\scriptsize PY};
\draw (n47) -- ++(135:0.6cm) node[fill=white, inner sep = 0pt] {\scriptsize JK};
\draw (n128) -- ++(180:0.8cm) node[fill=white, inner sep = 0pt] {\scriptsize PD};
\draw (n165) -- ++(180:0.8cm) node[fill=white, inner sep = 0pt] {\scriptsize VK};
\draw (n237) -- ++(160:1.2cm) node[fill=white, inner sep = 0pt] {\scriptsize HPK};
\end{tikzpicture}
\begin{tikzpicture}[scale = 0.4]
\tikzstyle{lm1} = [circle, fill = yafcolor8!50!black, inner sep = 1pt]
\tikzstyle{lm2} = [circle, fill = gray!50, inner sep = 1pt]
\tikzstyle{ll} = []
\tikzstyle{le1} = [yafcolor8, thick, bend left = 10]
\tikzstyle{le2} = [gray!50, bend left = 10, opacity=0.5]
\input{dblp_graph/dblp_whole}
\end{tikzpicture}

\caption{Discovered strong edges of 5 ego-networks of KDD innovation award winners.
The first 5 figures contain only strong edges: the colored edges and vertices show 5 topics
that were used as input: cluster, classif, pattern, network, distribut.
The last topic consisted of 2 connected components which we used as two separated communities.
The last figure shows strong \emph{and} weak edges.
Some of the vertices do no belong to any of the communities. Some edges are strong despite
not belonging to any of the communities because we keep edges that do not induce violations.
}
\label{fig:dblp}

\end{figure*}

\section{Concluding remarks}
\label{sec:conclusions}

We presented a novel approach for the problem
of inferring the strength of social ties. 
We assume that social ties can be one of two types, 
strong or weak, 
and as a guiding principle for the inference process 
we use the strong triadic closure property. 
In contrast to most works that use interaction data between users, 
which are private and thus, typically not available, 
we also consider as input a collection of {\em tight} communities.
Our assumption is that such tight communities 
are connected via strong ties. 
This assumption is valid in cases when 
being part of a community implies a strong connection to one of the existing members.
For instance, in a scientific collaboration network, 
a student is introduced on a research topic
by his/her supervisor who is already working in that topic.

Based on the \stc principle
and our assumption about com\-munity-level connectivity, 
we formulate two variants of the tie-strength inference problem:
\prbminstc, where we ask to minimize the number of \stc violations, 
and \prbmaxtri, where we the goal is to maximize the number of non-violated open triangles.
We show that both problems are \NP-hard.
Furthermore, we show that the \prbminstc problem  is hard to approximate, 
while for \prbmaxtri we develop an algorithm with approximation guarantee.
For the approximation algorithm we use a greedy algorithm
for maximizing a submodular function on intersection of matroids. 

There are many interesting directions to explore in the future. 
An interesting question is to consider alternative problem formulations that 
combine the strong triadic closure property with other community-level constraints, 
such as density and small diameter. 
We would also like to consider formulations that incorporate user features.
A different direction is to consider an interactive version of the problem, 
where the goal is select a small number of edges to query, 
so that the correct labeling on those edges can be used to 
maximize the accuracy of inferring the strength of the remaining edges.

\spara{Acknowledgements.}
This work was supported by 
the Tekes project ``Re:Know,'' 
the Academy of Finland project ``Nestor'' (286211),
and the EC H2020 RIA project ``SoBigData'' (654024).
\balance

\bibliographystyle{ACM-Reference-Format}
\bibliography{bibliography} 

\end{document}